\documentclass[reqno, 10pt]{amsart}

\usepackage[top=4cm, bottom=3.5cm, left=3.5cm, right=3.5cm]{geometry}

\newtheorem{thm}{Theorem}

\newtheorem{lem}[thm]{Lemma}

\theoremstyle{definition}
\newtheorem{defn}[thm]{Definition}
\theoremstyle{remark}
\newtheorem{rem}[thm]{Remark}

\numberwithin{equation}{section} 
\numberwithin{thm}{section}

\newcommand{\infspec}{{\rm inf\ spec\ }}

\title{Ground State Energy of Dilute Bose Gas in Small Negative Potential Case}

\author{Ji Oon Lee}

\address{Department of Mathematics, Harvard University, Cambridge, MA
02138, USA}

\email{jioon@math.harvard.edu}

\date{\today}

\begin{document}

\begin{abstract}
It is well known that the ground state energy of a three dimensional dilute Bose gas in the thermodynamic limit is $E=4\pi a \rho N$ when the particles interact via a non-negative, finite range, spherically symmetric, two-body potential. Here, $N$ is the number of particles, $\rho$ is the density of the gas, and $a$ is the scattering length of the potential. In this paper, we prove the same result without the non-negativity condition on the potential, provided the negative part is small.
\end{abstract}

\maketitle \thispagestyle{empty} \thispagestyle{headings}

\section{Introduction}

The ground state of a dilute Bose gas yields many interesting results. Dyson's estimate \cite{D} of the ground state energy was the first rigorous one, and about 40 years later, Lieb and Yngvason \cite{LY1} proved the correct leading term for the lower bound. Similar results for two-dimensional case were also proved subsequently \cite{LY2}.

All these works, however, were done under the assumption of non-negative interaction potential. In this paper, we relax this condition by permitting a small negative component in the potential. To see this more precisely, we introduce the following model.

We consider a system of $N$ interacting three-dimensional Bosons in a three dimensional torus $\mathbb{T}$ of side length $L$. Given the two-particle interaction $V$, the Hamiltonian of this system is
\begin{equation} \label{eqn0}
H = -\sum_{j=1}^N \Delta_j + \sum_{i<j}^N V(x_i - x_j)
\end{equation}
where $x_i \in \mathbb{T}$ are the positions of the particles and $\Delta_i$ denotes the Laplacian with respect to $i$-th particle. Every Bosonic state function in this paper is symmetric, smooth, and $L^2$-normalized. The potential $V$ is spherically symmetric, continuous, and compactly supported.

The ground state energy is defined as follows. For simplicity, we let $X_n=(x_1, \cdots, x_n)$ and $dX_n = dx_1 \cdots dx_n$.
\begin{defn}[The ground state energy]
For given Hamiltonian
\begin{equation}
H^{(n)}= -\sum_{j=1}^n \Delta_j + \sum_{i<j}^n V(x_i -x_j), \label{h}
\end{equation}
a) In a bounded set $S \subset \mathbb{R}^3$, its ground state energy, $E(n, S, V)$, with Neumann boundary conditions is defined to be
\begin{equation}
E(n, S, V) = \inf_{\|\psi\|_2 =1} \Big[ \sum_{j=1}^n \int_{S^n} dX_n |\nabla_j \psi(X_n)|^2 + \sum_{i<j}^n \int_{S^n} dX_n V(x_i - x_j) |\psi(X_n)|^2 \Big].
\end{equation}
b) In a three-dimensional torus of side length $s$, $\mathbb{T}_s = \mathbb{R}^3 /(s \mathbb{Z})^3$, its ground state energy $E(n, \mathbb{T}_s, V)$ is defined to be
\begin{equation}
E(n, \mathbb{T}_s, V) = \inf_{\|\psi\|_2 =1} \Big[ \sum_{j=1}^n \int_{\mathbb{T}_s^n} dX_n |\nabla_j \psi(X_n)|^2 + \sum_{i<j}^n \int_{\mathbb{T}_s^n} dX_n V(x_i - x_j) |\psi(X_n)|^2 \Big].
\end{equation}
\end{defn}

We want to find the ground state energy of (\ref{eqn0}) in $\mathbb{T} = \mathbb{T}_L$ in the thermodynamic limit, where $N$ and $L$ approach infinity with $\rho=N/L^3$ fixed. When $V$ is non-negative, it is known that the leading term of the ground state energy in $\Lambda$, a box of side length $L$, is $4 \pi a \rho N$ as $\rho \to 0$ \cite{LSSY}. To prove it, $\Lambda$ is divided into small boxes. Each smaller box, $\Lambda_{\ell}$ is of side length $\ell$ and big enough to have many particles in it. Then, the superadditivity of the ground state energy \cite{LY1},
\begin{equation}
E(n+n',\Lambda_{\ell}, V) \geq E(n,\Lambda_{\ell}, V) + E(n',\Lambda_{\ell}, V) \label{eqn1}
\end{equation}
tells that the infimum of the sum of the energies in the small boxes is attained when the particles are evenly distributed among those boxes. The superadditivity (\ref{eqn1}) results from the fact that we can neglect the interactions between $n$ particles and $n'$ particles when finding the lower bound, provided $V$ is non-negative.

Without the non-negativity, however, it is not true in general that adding particles in a box increases the energy. Moreover, it is clear that some negative potential is catastrophic in the sense that no energy lower bound exist. This fact suggests that we need to introduce some other conditions to ensure stability \cite{R}.

\begin{defn}[Stability of Potential]
A two-particle potential $V_0$ is stable if there exists $B \geq 0$ such that
\begin{equation}
\sum_{i<j}^n V_0(x_i - x_j) \geq -nB \label{stable}
\end{equation}
for all $n \geq 0$ and $x_1, x_2, \cdots, x_n \in \mathbb{R}^3$
\end{defn}

Now we can state the main theorem.

\begin{thm} \label{thm1}
Let $V_1$ and $V_2$ be non-negative, spherically symmetric, continuous, compactly supported, two-particle potentials, satisfying
\begin{equation}
V_1 (x) = 0 \;\;\;\text{if}\;\; |x|>R_0, \label{v1}
\end{equation}
\begin{equation}
V_2 (x) = 0 \;\;\;\text{if}\;\; |x|<R_0 \; \text{or} \; |x|>R_1. \label{v2}
\end{equation}
If $V_0 = V_1 - V_2$ is stable, i.e. satisfies (\ref{stable}), there exists a small positive constant $\lambda$ such that, if $V=V_1 - \lambda V_2$ and $a$ is the scattering length of $V$, then there are positive constants $C_0$ and $\epsilon$ such that the ground state energy of (\ref{eqn0}), with $N$ and $L$ large, has a lower bound
\begin{equation}
E(N, \mathbb{T}, V) \geq 4 \pi a \rho N (1-C_0 \rho^{\epsilon})
\end{equation}
as $\rho \rightarrow 0$.
\end{thm}

\begin{rem}
From the proof, one can estimate that $\epsilon < \frac{1}{31}$ and $C_0 > 9$ are sufficient for Theorem \ref{thm1}. It should also be noted that, however, the error term $C_0 \rho^{\epsilon}$ has no significant meaning. $\lambda$ is independent of $\rho$ and can be estimated from (\ref{delta}), (\ref{c1}), and (\ref{lambda}).
\end{rem}

\begin{rem}
Corresponding upper bound can be easily obtained with a proper assumption on $\lambda$. One can follow the proof of Theorem 2.2 in \cite{LSSY}. For detailed calculation, see \cite{D} and \cite{LSY}. (Remark 3.1 in \cite{LSY} explains how to include small negative potential in hard core potential, and one can prove the same result for the potential in Theorem \ref{thm1} using the similar argument in the remark.) More general upper bound calculation for a case of the interaction potential with positive scattering length can be found in \cite{Y}, which also considers the problem similar to the one in this paper, yet with a different approach.
\end{rem}

\begin{rem}
Throughout the paper, $C$ denotes various constants that do not depend on $\rho$.
\end{rem}

\vskip20pt

\section{Outline of the Proof}
\emph{Step 1: Non-negativity of the ground state energy in a small box of size $O(1)$. (Lemma \ref{lem1}.)}

We first choose the side length of the small box, $\ell$, and show that the ground state energy $E(2, S, V_1)$ is bounded below by a function of $\ell$ as in (\ref{e2}) for any rectangular box $S \subset \Lambda_{\ell}$. To control $V_2$, the negative potential, we use
\begin{equation}
E(k, S, V_1 -\delta V_2) \geq (1-\delta) E(k, S, V_1) + \delta E(k, S, V_1 - V_2).
\end{equation}
We have $E(k, S, V_1 - V_2 ) \geq -Bk$ due to the stability of the potential, and
\begin{equation}
E(k, S, V_1) \geq \lfloor \frac{k}{2} \rfloor E(2, S, V_1) \geq \frac{k}{3} E(2, S, V_1)
\end{equation}
from the superadditivity. These show that $E(k, S, V_1 - \delta V_2 ) \geq 0$ if we let $\delta$ small enough, which depends on $\ell$.

\vskip10pt

\emph{Step 2: Non-negativity of the ground state energy in any large boxes. (Theorem \ref{thm2}.)}

We divide the large box into small boxes. For a positive potential case, the theorem is trivial, since the energy decreases if we neglect interaction between particles in different boxes. In the negative potential case, however, this may increase energy.

To resolve the problem, we first assume that we have a torus and change the origin of division continuously, hence consider a set of divisions. In each division, we consider interactions between the particles in the same box only. This converts the original Hamiltonian $H$ into the sum of Hamiltonians in small boxes, $H_u$, where $u$ denotes the origin of division. We calculate the average of this sum $\int H_u$ for the whole set of division and prove that $H \geq \int H_u$. This average gives us a good estimate for the energy in the torus, and we can prove the non-negativity of the ground state energy in any large tori.

To prove the theorem for a box with Neumann boundary conditions, we notice that the only difference between a torus and a box comes from that, when we move the origin of the division, we have more smaller boxes, whose sizes vary, at the boundary of the large box for a box case. This does not cause any problem, however, because we saw in step 1 that the ground state energy is non-negative whenever the box is small enough to be contained in the small box $\Lambda_{\ell}$.

\vskip10pt

\emph{Step 3: Lower bound for the ground state energy of two particles in $\Lambda_{\ell_1}$. (Lemma \ref{lem5}.)}

We first estimate the ground state energy of Neumann problem (Lemma \ref{lem4}.)
\begin{equation}
(-\Delta + \frac{1}{2} V ) \phi = E \phi
\end{equation}
on a sphere of radius $\ell_0 \sim \rho^{-\frac{1}{3} + \frac{5 \epsilon}{3} }$. Using the perturbation theory (Lemma \ref{perturb}.), we find a lower bound for the ground state energy of two particles in $\Lambda_{\ell_1}$, a cubic box of side length $\ell_1 \sim \rho^{-\frac{1}{3} + \frac{\epsilon}{3}}.$

\vskip10pt

\emph{Step 4: Dividing and subdividing $\mathbb{T}$ into small boxes. (Lemma \ref{lem6}.)}

We divide $\mathbb{T}$ into small boxes of side length $\ell_2$ and find a lower bound for the ground state energy in the small box $\Lambda_{\ell_2}$. To find the lower bound, we subdivide the small box into smaller boxes of side length $\ell_1 \gg \ell_0$, $\Lambda_{\ell_1}$, and replace $V$ by the one that only has interaction when two particles are in the same box. Note that these division and subdivision include the technique of changing origin in Step 2.

We let $\ell_1$ small so that average number of particles in $\Lambda_{\ell_1}$ is much less than one. Furthermore, we neglect the energy in $\Lambda_{\ell_1}$ whenever $\Lambda_{\ell_1}$ contains only one particle or three or more particles; this does not increase the total energy in $\Lambda_{\ell_2}$ due to the non-negativity of the ground state energy in Step 2. To actually find the lower bound, we use the perturbation theory and estimate the number of cases in which exactly two particles are in the same $\Lambda_{\ell_1}$, when using the constant function, which is the ground state of the unperturbed Hamiltonian.

\vskip10pt

\emph{Step 5: Lower bound for the ground state energy in $\mathbb{T}$. (Theorem \ref{thm1}, Main Theorem.)}

Using superadditivity and the ground state energy and the ground state energy in $\Lambda_{\ell_2}$, which we obatined in Step 4, we can find lower bound for the ground state energy in $\mathbb{T}$. Here, we again consider interactions only when two particles are in the same box. The actual calculation gives the lower bound of the main theorem.

\vskip25pt

\section{Nonnegativity of Ground State Energy}
In this section, we prove the following theorem.

\begin{thm} \label{thm2}
For any $\ell'$, there exists a positive constant $c_1$ such that the ground state energy of (\ref{h}), where $V=V_1 - c_1 V_2 $ and $V_1$ and $V_2$ satisfy the assumptions of Theorem \ref{thm1}, in $\mathbb{T}_{\ell'}$, a three dimensional torus of size $\ell'$, or in $\Lambda_{\ell'}$, a box of size $\ell'$, is non-negative for any $n$, i.e.
\begin{equation}
E(n, \mathbb{T}_{\ell'}, V_1 -c_1 V_2) \geq 0,
\end{equation}
and
\begin{equation}
E(n, \Lambda_{\ell'}, V_1 -c_1 V_2) \geq 0.
\end{equation}
\end{thm}

\vskip10pt

We first prove this for a small cell $\Lambda_{\ell}$ of side length $\ell \sim O(1)$, which contains $k$ particles. To begin with, we first show a lemma that we will use throughout this paper.

\begin{lem} \label{perturb}
Let $A$ be a non-negative Hermitian operator on $L^2 (S)$, where $S$ is a bounded set. Assume that the constant function $\psi_0$ is an eigenfunction of $A$ with a simple eigenvalue $0$ and $\|\psi_0\|_{L^2 (S)} =1$. Furthermore, assume that the second smallest eigenvalue, the gap, of $A$ is $\gamma$. Suppose that $X$ is a multiplication operator on $L^2 (S)$ and let $X_{\infty} = \|X\|_{L^{\infty}(S)}$. Then, if $\gamma \geq 4 X_{\infty}$,
\begin{equation}
\infspec (A+X) \geq \langle \psi_0, X\psi_0 \rangle - \frac{2|X_{\infty}|^2} {\gamma}
\end{equation}
\end{lem}

\begin{proof}[Proof of Lemma \ref{perturb}.]
Let $E_0 = \infspec (A+X)$. Then,
\begin{equation}
E_0 \leq \langle \psi_0, (A+X) \psi_0 \rangle = \langle \psi_0, X \psi_0 \rangle \leq X_{\infty}.
\end{equation}
Let $\psi_0 + \psi'$ be the eigenfunction of $A+X$ with the smallest eigenvalue and $\langle \psi_0, \psi' \rangle =0$. Then,
\begin{equation}
(A+X)(\psi_0 + \psi') = E_0 (\psi_0 + \psi') \label{e0}
\end{equation}
Taking inner products of both sides of (\ref{e0}) with $\psi_0$ and $\psi'$ gives
\begin{equation}
\langle \psi_0, X(\psi_0 + \psi') \rangle = E_0, \label{inner1}
\end{equation}
\begin{equation}
\langle \psi', X(\psi_0 + \psi' ) \rangle + \langle \psi', A \psi' \rangle = E_0 \langle \psi', \psi' \rangle,
\end{equation}
respectively. By the assumption, $\langle \psi', A\psi' \rangle \geq \gamma \langle \psi', \psi' \rangle$, thus
\begin{equation}
E_0 \langle \psi', \psi' \rangle \geq \langle \psi', X(\psi_0 + \psi') \rangle + \gamma \langle \psi', \psi' \rangle \label{perturb1}
\end{equation}

From the Schwarz inequality, we have
\begin{eqnarray}
\nonumber \langle \psi', X(\psi_0 + \psi') \rangle &\geq& - X_{\infty} \|\psi'\|_{L^2 (S)} \|\psi_0 + \psi' \|_{L^2 (S)} = - X_{\infty} \|\psi'\|_{L^2 (S)} (1+ \|\psi'\|_{L^2 (S)}^2 )^{\frac{1}{2}} \\
&\geq& - X_{\infty} ( \|\psi'\|_{L^2 (S)} + \|\psi'\|_{L^2 (S)}^2 ) \label{sch}.
\end{eqnarray}
Thus, from (\ref{perturb1}) and (\ref{sch}),
\begin{equation}
(E_0 - \gamma) \|\psi'\|_{L^2 (S)}^2 \geq - X_{\infty}(\|\psi'\|_{L^2 (S)} + \|\psi'\|_{L^2 (S)}^2),
\end{equation}
and, since $E_0 \leq X_{\infty} \leq \gamma$,
\begin{equation}
\|\psi'\|_{L^2 (S)} \leq \frac{X_{\infty}}{\gamma - E_0 - X_{\infty}}.
\end{equation}
Therefore, from (\ref{inner1}),
\begin{eqnarray}
\nonumber E_0 &=& \langle \psi_0, X\psi_0 \rangle + \langle \psi_0, X \psi' \rangle \\
&\geq& \langle \psi_0, X\psi_0 \rangle - X_{\infty} \|\psi'\|_{L^2 (B)} \geq \langle \psi_0, X\psi_0 \rangle - \frac{|X_{\infty}|^2}{\gamma - E_0 - X_{\infty}} \\
\nonumber &\geq& \langle \psi_0, X\psi_0 \rangle - \frac{2 |X_{\infty}|^2} {\gamma},
\end{eqnarray}
which was to be proved.
\end{proof}

Now we prove Theorem \ref{thm2} for a small box $\Lambda_{\ell}$ of side length $\ell$.

\begin{lem} \label{lem1}
Let $\Lambda_{\ell}$ be a box of side length $\ell$. Then, there exists $\delta > 0$ such that for any rectangular box $S \subset \Lambda_{\ell}$ and for any $k \geq 2$, the ground state energy of
\begin{equation}
H^{(k)}= -\sum_{j=1}^k \Delta_j + \sum_{i<j}^k (V_1 - \delta V_2 ) (x_i -x_j)
\end{equation}
in $S$ with Neumann boundary conditions is non-negative, i.e.
\begin{equation}
E(k, S, V_1 - \delta V_2) \geq 0
\end{equation}
Here, $R_1$ is the range of the potential $V_1 - \delta V_2$ as in Theorem \ref{thm1}.
\end{lem}

\begin{proof}[Proof of Lemma \ref{lem1}.]
We first prove that $E(2, S, V_1) \geq E' > 0$. Since $V_1 - V_2$ is stable, $V_1(0) >0$. Thus, we can find $R>0$ such that $V_1(x) >\frac{V_1(0)}{2}$ if $|x| < R$. Let
\begin{equation}
V_1'(x) =
  \begin{cases}
  V_1(x) - \frac{V_1 (0)}{2} & \text{if} \;\; |x| <R \\
  0 & \text{otherwise}
  \end{cases},
\end{equation}
and $a'$ be the scattering length of $2 V_1'$.

By change of variable $\xi = x_1 + x_2$, $\eta = x_1 - x_2$, and $S'=\{(\xi, \eta) | \frac{\xi + \eta}{2}, \frac{\xi - \eta}{2} \in \Lambda_{\ell} \}$, we get
\begin{eqnarray}
\nonumber && \int_{S^2} dx_1 dx_2 \big[ |\nabla_1 \psi |^2 + |\nabla_2 \psi |^2 + V_1 (x_1 - x_2 ) |\psi |^2 \big]  = \frac{1}{8} \int_{S^2} d \xi d\eta \big[ |\nabla_{\xi} \psi |^2 + |\nabla_{\eta} \psi |^2 + V_1 (\eta ) |\psi |^2 \big]  \\
&\geq& \frac{1}{8} \int_{S'} d \xi d \eta  \big[ |\nabla_{\eta} \psi |^2  + V_1' (\eta) |\psi |^2 + \frac{V_1(0)}{2} 1(|\eta| <R ) |\psi|^2 \big].
\end{eqnarray}
To find the lower bound, we use the following lemma, the generalization of a Lemma of Dyson \cite{LY1}.

\begin{lem} \label{dyson}
Let $v(r) \geq 0$ with scattering length $a$ and $v(r) = 0$ for $r > R_0$. Let $U(r) \geq 0$ be any function satisfying $\int U(r)r^2 dr \leq 1$ and $U(r) = 0$ for $r < R_0$. Let $D \subset \mathbb{R}^3$ be star-shaped (convex suffices) with respect to $0$. Then, for all differentiable functions $\phi$,
\begin{equation}
\int_D dx \Big[ |\nabla \phi (x) |^2 + \frac{1}{2} v(r) |\phi(x)|^2 \Big] \geq a \int_D dx U(r) |\phi(x)|^2.
\end{equation}
\end{lem}
\begin{proof}
See Lemma 1 in \cite{LY1}.
\end{proof}

Choose $\tilde{\ell}$ large so that
\begin{equation} \label{tilde_ell}
0< \frac{3a'}{(2\tilde{\ell})^3 - R^3} < \frac{V_1 (0)}{2}.
\end{equation}
Define $U_{\ell} (r)$ as follows: When $\ell \geq \tilde{\ell}$,
\begin{equation}
U_{\ell} (r) =
  \begin{cases}
  3[ (2\ell)^3 - R^3 ]^{-1} & \text{if} \;\; R<r<2 \ell \\
  0 & \text{otherwise}
  \end{cases},
\end{equation}
and, when $\ell < \tilde{\ell}$,
\begin{equation}
U_{\ell} (r) =
  \begin{cases}
  3[ (2\tilde{\ell})^3 - R^3 ]^{-1} & \text{if} \;\; R<r<2 \ell \\
  0 & \text{otherwise}
  \end{cases},
\end{equation}

For any fixed $\xi$, let $S_{\xi} = \{ \eta | \frac{\xi + \eta}{2}, \frac{\xi - \eta}{2} \in \Lambda_{\ell} \}$. Then, $\int U(r) r^2 dr = 1$. By Lemma \ref{dyson}, since $S_{\xi}$ is convex,
\begin{equation}
\int_{S_{\xi}} d\eta \big[ |\nabla_{\eta} \psi |^2 + V_1 ' (\eta ) |\psi|^2 \big] \geq a' \int_{S_{\xi}} d\eta U_{\ell}(\eta) | \psi |^2.
\end{equation}
Let
\begin{equation}
E'(\ell) := 
  \begin{cases}
  \displaystyle \frac{3a'}{(2\tilde{\ell})^3 - R^3} & \text{if} \;\; \ell < \tilde{\ell} \\
  \displaystyle \frac{3a'}{(2\ell)^3 - R^3} & \text{if} \;\; \ell \geq \tilde{\ell}
  \end{cases}.
\end{equation}
Then, we get
\begin{eqnarray}
\int_{S_{\xi}} d \eta  \big[ |\nabla_{\eta} \psi |^2  + V_1' (\eta) |\psi |^2 + \frac{V_1(0)}{2} 1(|\eta| <R ) |\psi|^2 \big] \geq E'(\ell) \int_{S_{\xi}} d\eta | \psi |^2.
\end{eqnarray}
Hence,
\begin{equation}
E(2, S, V_1) = \inf_{\| \psi \|_2 =1} \int_{S^2} dx_1 dx_2 \Big[ |\nabla_1 \psi |^2 + |\nabla_2 \psi |^2 + V_1 (x_1 - x_2 ) |\psi|^2 \Big] \geq E'. \label{e2}
\end{equation}

Now, from the superadditivity (\ref{eqn1}),
\begin{equation}
E(k, S, V_1 ) \geq \lfloor \frac{k}{2} \rfloor E(2, S, V_1) \geq \frac{k}{3} E(2, S, V_1 ) \geq \frac{E' k}{3},
\end{equation}
where $\lfloor \frac{k}{2} \rfloor$ denotes the greatest integer that does not exceed $\frac{k}{2}$. Hence, if
\begin{equation}
\delta \leq \min \{ \frac{1}{2}, \frac{E'}{6B} \}, \label{delta}
\end{equation}
then,
\begin{eqnarray}
\nonumber &&\sum_{j=1}^k \int_{S^k} dX_k |\nabla_j \psi(X_k)|^2 + \sum_{i<j}^k \int_{S^k} dX_k \big( V_1(x_i-x_j)-\delta V_2(x_i-x_j) \big) |\psi(X_k)|^2 \\
\nonumber &\geq& (1-\delta) \Big[ \sum_{j=1}^k \int_{S^k} dX_k |\nabla_j \psi(X_k)|^2 + \sum_{i<j}^k \int_{S^k} dX_k V_1(x_i-x_j) |\psi(X_k)|^2 \Big]\\
&& +\;\;\delta \sum_{i<j}^k \int_{S^k} dX_k \big(V_1(x_i-x_j)- V_2(x_i-x_j) \big) |\psi(X_k)|^2 \\
\nonumber &\geq& \frac{1}{2}\cdot \frac{E' k}{3} -\delta B k \geq 0.
\end{eqnarray}
This proves the lemma.
\end{proof}

Lemma \ref{lem1} proves Theorem \ref{thm2} when $\ell' = O(1)$. Since $\delta$ depends on $\ell$, however, the argument we used in the proof of Lemma \ref{lem1} is no longer valid for a fixed $\delta$ when $\ell'$ is large.

In the case when $\ell'$ is large, we divide the box of size $\ell'$ into smaller boxes. However, the potential here contains small negative parts in it, so that dividing the box does not guarantee that the ground state energy decreases. Thus, using the argument similar to \cite{CLY}, we change the grid for division continuously and take an average of the ground state energy for such divisions.

\begin{proof}[Proof of Theorem \ref{thm2}.]
We first prove the torus case. Let $\ell = \max \{ 10 \tilde{\ell}, 2 R_1 \}$ where $\tilde{\ell}$ is defined as in \eqref{tilde_ell} in the proof of Lemma \ref{lem1}. Let $\ell' \geq \ell$. Let $\chi$ be a characteristic function of the unit cube,
\begin{equation}
\chi(x)=
\begin{cases}
1 & x \in [0, 1]^3 \\
0 & \text{otherwise}\\
\end{cases},
\end{equation}
defined on $\mathbb{T}_{\ell'} = \mathbb{R}^3 / (\ell' \mathbb{Z})^3$.

Let $\chi_{u \lambda}(x) = \chi (x+u+\lambda)$. Here, $\lambda \in G$ where $G=\{p \in \mathbb{Z}^3 | \ell p \in \Lambda_{\ell'} \}$, and $u \in [0, 1]^3 \equiv \Gamma$. Then, $\sum_{\lambda \in G} \chi_{u \lambda} (x) =1$ for all $x \in \mathbb{T}_{\ell'}, u \in \Gamma$.

A function $h$ is defined by
\begin{equation}
h(x,y) = \int_{\Gamma} du \sum_{\lambda \in G} \chi_{u \lambda} (x) \chi_{u \lambda} (y). \label{hl}
\end{equation}
Then, $h$ depends only on the difference $z=x-y$ and we can let
\begin{eqnarray}
\nonumber h(z) &=& h(x, y) = \int_{\Gamma} du \sum_{\lambda \in G} \chi (x+u+\lambda) \chi (y+u+\lambda) \\
&=& \int_{\mathbb{T}_{\ell'}} du \;\chi (x+u) \chi (y+u) = (\chi * \chi )(z).
\end{eqnarray}

Let $h_{\ell} (x) =h(x/\ell)$. We can define localized kinetic and potential energies. Let $\alpha = (u, \alpha_1 , \cdots, \alpha_n) \in \Gamma \times G^n $ be a multi-index and $\int d\alpha =
\int_{\Gamma} du \sum_{\alpha_1 \in G} \cdots \sum_{\alpha_n \in G}$. Let
\begin{equation}
\psi_{\alpha}^{\ell}(x_1, \cdots, x_n) = \prod_{k=1}^{n} \chi_{u \alpha_k} (\frac{x_k}{\ell}) \psi (x_1, \cdots, x_n)
\end{equation}
\begin{equation}
V_{\alpha} = \sum_{i<j}^n \delta_{\alpha_i \alpha_j} V,
\end{equation}
where $V=V_1 - \delta V_2$. Then from the definition of $\chi$ and $h$, we have
\begin{eqnarray}
\nonumber \int d\alpha  V_{\alpha} |\psi_{\alpha}^{\ell}|^2 &=& \sum_{i<j}^{n} \int_{\mathbb{T}_{\ell'}^n} dX_n \int d\alpha \prod_{k=1}^n \chi_{u \alpha_k} (\frac{x_k}{\ell}) V(x_i -x_j) \delta_{\alpha_i \alpha_j} |\psi(X_n )|^2 \\
&=& \sum_{i<j}^n \int_{\mathbb{T}_{\ell'}^n} dX_n V(x_i -x_j) h_{\ell} (x_i -x_j ) |\psi(X_n)|^2. \label{kinetic}
\end{eqnarray}

We also know that
\begin{eqnarray}
\nonumber & &\sum_{j=1}^n \int_{\mathbb{T}_{\ell'}^n} dX_n|\nabla_j \psi(X_n)|^2 \\
&=&\int_{\Gamma} du \sum_{\sum n_{\sigma} =n} \int_{\mathbb{T}_{\ell'}^n} dX_n \sum_{j=1}^n |\nabla_j \psi(X_n)|^2 \prod_{\sigma} 1(N_{u \sigma} (X_n) =n_{\sigma}) \label{potential} \\
\nonumber &=&\sum_{\sum N_{\sigma} = n} \int d\alpha \int_{\mathbb{T}_{\ell'}^n} dX_n \prod_{k=1}^n \chi_{u\alpha_k} (\frac{x_k}{\ell}) \sum_{j=1}^n |\nabla_j \psi(X_n)|^2 \prod_{\sigma} 1(N_{u \sigma} (X_n) =n_{\sigma} )
\end{eqnarray}
where $N_{u \sigma}(X_n) = \sum_{j=1}^n \chi_{u \sigma}(x_j)$.

Thus, from the equations above, we get the following lemma.

\begin{lem} \label{lem7}
For any $n$,
\begin{equation}
E(n, \mathbb{T}_{\ell'}, V h_{\ell}) \geq \inf_{ \sum n_{\sigma} =n} \sum_{\sigma \in G} E(n_{\sigma}, \Lambda_{\ell}, V)
\end{equation}
\end{lem}

\begin{proof}[Proof of Lemma \ref{lem7}]
Combining \eqref{kinetic} and \eqref{potential}, we get
\begin{eqnarray}
\nonumber && \sum_{j=1}^n \int_{\mathbb{T}_{\ell'}^n} dX_n |\nabla_j \psi(X_n)|^2 + \sum_{i<j}^n \int_{\mathbb{T}_{\ell'}^n} dX_n \; V(x_i -x_j) h_{\ell} (x_i -x_j ) |\psi(X_n)|^2 \\
&\geq& \Bigl( \int d\alpha \| \psi_{\alpha}^{\ell} \|_2^2 \Bigr) \inf_{ \sum n_{\sigma} =n} \sum_{\sigma \in G} E(n_{\sigma}, \Lambda_{\ell}, V) = \inf_{ \sum n_{\sigma} =n} \sum_{\sigma \in G}
E(n_{\sigma}, \Lambda_{\ell}, V)
\end{eqnarray}
\end{proof}

Recall that $V_2 (x_i -x_j )=0$ when $|x_i -x_j| \geq R_1$. Since $R_1=O(1)$ and $h(z) = (\chi * \chi)(z) = g(z^1)g(z^2)g(z^3)$ where $z=(z^1, z^2, z^3)$ and
\begin{equation}
g(t)=
\begin{cases}
1-|t| & |t| \leq 1 \\
0 & \text{otherwise}\\
\end{cases},
\end{equation}
we can let $\ell \geq 2 R_1$ so that $|x_i -x_j| <R_1$ implies $h_{\ell} (x_i -x_j) \geq 1-\frac{\sqrt{3} R_1}{\ell}$. Now, if we let
\begin{equation}
c_1 =(1-\frac{\sqrt{3} R_1}{\ell}) \delta, \label{c1}
\end{equation}
where $\delta$ is defined in Lemma \ref{lem1} with $\ell$, then $V_1 \geq h_{\ell}V_1$ and $c_1 V_2 \leq \delta V_2 h_{\ell}$, hence
\begin{eqnarray}
\nonumber &&\sum_{j=1}^n \int_{\mathbb{T}_{\ell'}} dX_n |\nabla_j \psi(X_n)|^2 + \sum_{i<j}^n \int_{\mathbb{T}_{\ell'}} dX_n \;(V_1 - c_1 V_2 )(x_i-x_j) |\psi(X_n)|^2 \\
&\geq& \sum_{j=1}^n \int_{\mathbb{T}_{\ell'}} dX_n |\nabla_j \psi(X_n)|^2 + \sum_{i<j}^n \int_{\mathbb{T}_{\ell'}} dX_n \;(V_1 - \delta V_2 )(x_i-x_j) h_{\ell}(x_i-x_j) |\psi(X_n)|^2.
\end{eqnarray}

Hence, if $E(n_{\sigma}, \Lambda_{\ell}, V) \geq 0$ for any $n_{\sigma}$, the ground state energy is non-negative, and it is already proved in Lemma \ref{lem1}. This proves the first part of the theorem.

\vskip10pt

The box case can be proved in a similar way. When we have a box of size $\ell'$, $\Lambda_{\ell'}$, the only difference between $\Lambda_{\ell'}$ and $\mathbb{T}_{\ell'}$ is that $\Lambda_{\ell'}$ has walls at the boundary, in the sense that the particles do not interact across these walls. Thus, if we use the same argument as in the case of the torus, we can get the Lemma \ref{lem7} with some small cells at the boundary of $\Lambda_{\ell'}$ having walls in them. If one of these small cells, $B$, is divided into $B_1, B_2, \cdots, B_m$ by the walls, then, in this case, $E(k, B, V)$ actually denotes $\displaystyle \inf_{\sum k_m = k} \sum_{j=1}^m E(k_j, B_j, V)$. Each $E(k_j, B_j, V)$ is the ground state energy in $B_j$, which is smaller than $B$, and Lemma \ref{lem1} holds for smaller cells. Thus, $E(k_j, B_j, V) \geq 0$ and $\displaystyle \inf_{\sum k_m = k} \sum_{j=1}^m E(k_j, B_j, V) \geq 0$. Hence, we get the desired theorem.
\end{proof}

\vskip30pt

\section{Lower Bound of Ground State Energy}
With the aid of the Theorem \ref{thm2}, we will prove our main result. Assuming conditions in Theorem \ref{thm1}, we start with a lemma which is similar to lemma A.1 in \cite{ESY}, with the interaction potential
\begin{equation}
V= V_1 - \frac{1}{2} c_1 V_2 = V_1 - \lambda V_2, \label{lambda}
\end{equation}
where $c_1$ satisfies Theorem \ref{thm2} and is defined by (\ref{delta}) and (\ref{c1}).

\begin{lem} \label{lem4}
Let $\phi$ be a solution of Neumann Problem
\begin{equation}
(-\Delta + \frac{1}{2}V) \phi = E\phi
\end{equation}
on the sphere of radius $\ell_0$, where E is the ground state energy, with the boundary condition
\begin{equation}
\phi(\ell_0) =1, \;\;\; \frac{\partial \phi}{\partial r}(\ell_0)=0.
\end{equation}
Let $a$ be the scattering length of $V$. Then if $\ell_0 \gg R_1$, we have
\begin{equation}
E \geq \frac{3a}{\ell_0^3} \bigl(1 + O(\frac{1}{\ell_0}) \bigr).
\end{equation}
\end{lem}

\begin{proof}[Proof of Lemma \ref{lem4}.]
Let $w$ be the solution of zero-energy scattering equation,
\begin{equation}
(-\Delta + \frac{1}{2}V) w = 0, \;\; \lim_{r \rightarrow \infty} w(r)=1.
\end{equation}
Note that $w(r) > 0$ even when $V(r) < 0$. Then, by definition, $f(r):=rw(r) =r-a$ for $r>R_1$ where $r=|x|$. Let
\begin{equation}
\psi(r) = \frac{\sin \big( hf(r) \big)}{r}.
\end{equation}
Here, $h$ is the smallest positive number satisfying $\psi'( \ell_0 ) = 0$, or $h = \ell_0^{-1} \tan \big( hf(\ell_0) \big)$. It gives
\begin{equation}
h^2 = \frac{3a}{\ell_0^3}+ O(\frac{1}{\ell_0^4}).
\end{equation}

Let $g = \psi^{-1} \phi$. $g$ is well-defined, since $\psi > 0$. Then
\begin{equation}
\int_{r<\ell_0} \overline{\phi} (-\Delta + \frac{1}{2}V )\phi = \int_{r<\ell_0} |g|^2 \psi (-\Delta + \frac{1}{2}V)\psi + \int_{r<\ell_0} |\nabla g|^2 \psi^2.
\end{equation}
A calculation shows that
\begin{eqnarray}
\nonumber & & \psi (-\Delta + \frac{1}{2}V)\psi \\
&=& \frac{1}{2} V\psi^2 + \psi \Bigl[\frac{h^2 \{f'(r)\}^2 \sin \big( h f(r) \big)}{r} - \frac{h f''(r) \cos \big( h f(r) \big)}{r} \Bigr] \\
\nonumber &=& h^2 \psi^2 + h^2 \psi^2 (\{f'(r)\}^2 -1) - \frac{1}{r^2} \bigl[h f''(r) \cos \big( h f(r) \big) \sin \big( h f(r) \big) + \frac{1}{2}V \sin^2 \big( h f(r) \big) \bigr]
\end{eqnarray}
The last two terms vanish where $r>R_1$. Furthermore, in the Taylor expansion with respect to $h$, $O(h^2)$ terms cancel in the square bracket, since $-f''(r)+\frac{1}{2}V(r)f(r)=0$. Thus, we have
\begin{eqnarray}
\nonumber && \int_{r<\ell_0} \overline{\phi}( -\Delta + \frac{1}{2} V) \phi \\
&\geq& h^2 \int_{r<\ell_0} |\phi|^2 + h^2 \int_{r<\ell_0} (\{f'(r)\}^2 -1) |\phi|^2 - C h^3 \int_{r<R_1} \frac{|g|^2}{r^2} + \int_{r<\ell_0} |\nabla g|^2 \psi^2.
\end{eqnarray}

Note that $\phi$ is bounded from above, since $V_2$ is continuous, and $\phi$ is bounded from below, since $V_1$ is continuous. (See Lemma A.1 in \cite{ESY} for properties of $\phi$.) Since $f'(r)=1$ for $r>R_1$ and $f'$ is bounded, we can see that $\int ({f'(r)}^2 -1) |\phi|^2 / \int |\phi|^2 = O(\ell_0^{-1})$ as $\ell_0 \rightarrow \infty$. We also know that $\sin \big(h f(r) \big) \geq C hr$ and $f$ does not vanish. Hence, $\psi \geq C h$ and
\begin{equation}
\int_{r<\ell_0} |\nabla g|^2 \psi^2 \geq C h^2 \int_{r<\ell_0} |\nabla g|^2.
\end{equation}

Now, from the Hardy type inequality (Lemma 5.1 in \cite{ESY}),
\begin{equation}
\int_{r<\ell_0} \frac{1(|x|<R_1)}{r^2} |g|^2 \leq C \int_{r<\ell_0} |\nabla g|^2 + C \frac{R_1^3}{\ell_0^3} \int_{r<\ell_0} |g|^2,
\end{equation}
where we used
\begin{equation}
|\{x \in \mathbb{R}^3 : |x|<\ell_0 \}|^{-1} \int_{r<\ell_0} 1(|x|<R_1) = \frac{R_1^3}{\ell_0^3}
\end{equation}

Thus,
\begin{eqnarray}
\nonumber & &\int_{r<\ell_0} |\nabla g|^2 \psi^2 - C h^3 \int_{r<R_1} \frac{|g|^2}{r^2} \geq C \Bigl(h^3 \int_{r<\ell_0} |\nabla g|^2 - h^3 \int_{r<R_1} \frac{|g|^2}{r^2} \Bigr) \\
&\geq& -C h^3 \frac{R_1^3}{\ell_0^3} \int_{r<\ell_0} |g|^2.
\end{eqnarray}

Since $h^2 = O(\ell_0^{-3})$ and $\psi^2 = O(\ell_0^{-3})$, the last term is bounded by $C h^3 \int |\phi|^2$, and we get the lemma.
\end{proof}

To prove the two-particle case, let $\ell_1^3 = \rho^{-1 + \epsilon}$, $\ell_0^3 = \rho^{-1 + 5 \epsilon}$, $\epsilon < \frac{1}{31}$. This choice of $\ell_1$ guarantees $\Lambda_{\ell_1}$, a box of side length $\ell_1 \gg \ell_0$, contains almost no particles in average.

\begin{lem} \label{lem5}
Let the two-particle Hamiltonian
\begin{equation}
H^{(2)} = -\Delta_1 - \Delta_2 + V(x_1 -x_2).
\end{equation}
Let the density of particles in the thermodynamic limit $\rho = N / L^3$, and $\ell_1^3 = \rho^{-1 + \epsilon}$. Then, there exists a positive constant $\epsilon$ such that its ground state energy with Neumann boundary conditions satisfies
\begin{equation}
E(2, \Lambda_{\ell_1}, V) \geq (1-\rho^{\epsilon}) \frac{8 \pi a}{\ell_1^3}.
\end{equation}
\end{lem}

\begin{proof}[Proof of Lemma \ref{lem5}.]
For two-particle bosonic function $\phi(x_1, x_2 )$ with Neumann boundary condition, we have
\begin{eqnarray}
\nonumber && \int_{\Lambda_{\ell_1}^2} dx_1 dx_2 \overline{\phi} H^{(2)} \phi \\
&=& \frac{\rho^{\epsilon}}{2} \int_{\Lambda_{\ell_1}^2} dx_1 dx_2 \overline{\phi}(-\Delta_1 - \Delta_2) \phi + (1-\rho^{\epsilon}) \int_{\Lambda_{\ell_1}^2} dx_1 dx_2 \nonumber \overline{\phi} \big(-\Delta_1 + \frac{1}{2}V(x_1 -x_2) \big) \phi \label{eqn_two}\\
&& + (1-\rho^{\epsilon}) \int_{\Lambda_{\ell_1}^2} dx_1 dx_2 \overline{\phi} \big( -\Delta_2 + \frac{1}{2}V(x_1 -x_2) \big) \phi \\
\nonumber && + \frac{\rho^{\epsilon}}{2} \int_{\Lambda_{\ell_1}^2} dx_1 dx_2 \overline{\phi}\big((-\Delta_1 - \Delta_2) + 2V(x_1-x_2) \big)\phi.
\end{eqnarray}

Let $\widetilde{\Lambda_{\ell_1}} := \{ p \in \Lambda_{\ell_1} : d(p, \partial \Lambda_{\ell_1} ) > \ell_0 \}$, where $d(p, \partial \Lambda_{\ell_1} )$ denotes the distance between a point $p$ and $\partial \Lambda_{\ell_1}$, the boundary of $\Lambda_{\ell_1}$. For $j=1, 2$, we want to show that
\begin{eqnarray}
&& \int_{\Lambda_{\ell_1}^2} dx_1 dx_2 \overline{\phi} \big( -\Delta_j + \frac{1}{2}V(x_1 -x_2) \big) \phi \label{H_1 estimate}\\
&\geq& \frac{3a}{\ell_0^3}\Big(1+O(\frac{1}{\ell_0})\Big) \int_{\Lambda_{\ell_1}^2} dx_1 dx_2 1(|x_1 -x_2 | < \ell_0) 1(x_1 \in \widetilde{\Lambda_{\ell_1}}) 1(x_2 \in \widetilde{\Lambda_{\ell_1}}) |\phi|^2, \nonumber
\end{eqnarray}
where $\ell_0^3 = \rho^{-1 + 5 \epsilon}$.

To see this, consider the case $j=1$ and let $x_2$ be fixed. If $d(x_2, \partial \Lambda_{\ell_1} ) > \ell_0$, then Lemma \ref{lem4} shows that
\begin{equation}
\int_{|x_1 - x_2|<\ell_0} dx_1 \overline{\phi} \big( -\Delta_1 + \frac{1}{2}V(x_1 -x_2) \big) \phi \geq \frac{3a}{\ell_0^3}\Big(1+O(\frac{1}{\ell_0})\Big) \int_{|x_1 - x_2|<\ell_0} dx_1 |\phi|^2.
\end{equation}

When $d(x_2, \partial \Lambda_{\ell_1} ) \leq \ell_0$, we only need to show that
\begin{equation}
\int_{\Lambda_{\ell_1}} dx_1 \overline{\phi} \big( -\Delta_1 + \frac{1}{2}V(x_1 -x_2) \big) \phi \cdot 1(|x_1 - x_2|<\ell_0) \geq 0.
\end{equation}
To prove this, we let, for $\varphi : (0, \infty) \to \mathbb{C}$ satisfying $\int_0^{\infty} r^2 |\varphi(r)|^2 dr < \infty$,
\begin{equation}
I_s (\varphi) := \int_0^s dr \Big( |\frac{\partial \varphi}{\partial r}|^2 + \frac{1}{2} V(r) |\varphi|^2 \Big) r^2.
\end{equation}
By the assumptions \eqref{v1} and \eqref{v2}, we can immediately see that $I_s (\varphi) \geq 0$ if $s < R_0$. Suppose that $I_s (\varphi) < 0$ for some $s \geq R_0$ and $\varphi$. We can assume that $\frac{d \varphi}{dr} (s) = 0$, since we can adjust $\varphi(r)$ near $r=s$ so that $\frac{d \varphi}{dr} (s) = 0$ and $I_s (\varphi) < 0$. Consider a function $\zeta \in L^2 (\mathbb{R}^3)$ defined by
\begin{equation}
\zeta (x) = 
	\begin{cases}
	\varphi(|x|) & \text{if} \;\;\; |x| < s \\
	\varphi(s) & \text{if} \;\;\; |x| \geq s
	\end{cases}.
\end{equation}
Then, since $V(x) \leq 0$ when $|x| \geq R_0$ by \eqref{v1} and \eqref{v2},
\begin{equation}
\int_{B(0, \ell_0)} \overline{\zeta} (-\Delta + \frac{1}{2} V ) \zeta < 0,
\end{equation}
which contradicts Lemma \ref{lem4}. This proves \eqref{H_1 estimate}.

From Lemma \ref{lem1}, we can see that the last term in (\ref{eqn_two}) is non-negative, since
\begin{eqnarray}
\nonumber && \int_{\Lambda_{\ell_1}^2} dx_1 dx_2 \overline{\phi}\big((-\Delta_1 - \Delta_2) + 2V(x_1-x_2) \big) \phi \\
&\geq& \int_{\Lambda_{\ell_1}^2} dx_1 dx_2 \overline{\phi}\big((-\Delta_1 - \Delta_2) + (V_1 - c_1 V_2)(x_1-x_2) \big) \phi \geq 0.
\end{eqnarray}
This proves
\begin{equation}
H^{(2)} \geq H_0 + H_1
\end{equation}
where
\begin{equation}
H_0 = \frac{\rho^{\epsilon}}{2} (-\Delta_1 - \Delta_2)
\end{equation}
and
\begin{equation}
H_1 = \frac{6a}{\ell_0^3} \Big(1+O(\frac{1}{\ell_0})\Big) (1-\rho^{\epsilon}) 1(|x_1 - x_2 | \leq \ell_0) 1(x_1 \in \widetilde{\Lambda_{\ell_1}}) 1(x_2 \in \widetilde{\Lambda_{\ell_1}}).
\end{equation}

We let $H_0$ be the unperturbed part and $H_1$ the perturbation in $H_0 + H_1$. Lemma \ref{perturb} tells that the ground state energy of $H^{(2)}$ satisfies
\begin{equation} \label{eqn_per}
E(2, \Lambda_{\ell_1}, V) \geq \int_{\Lambda_{\ell_1}^2} dx_1 dx_2 \overline{\psi_0} H_1 \psi_0 - \frac {2 \|H_1\|_{\infty}^2}{\gamma}
\end{equation}
where $\psi_0$ is the ground state of $H_0$, and $\gamma = C \rho^{\epsilon} \ell_1^{-2}$ is the spectral gap of $H_0$. We need the condition $\gamma \geq \|H_1\|_{\infty}$ and it indeed is true, since
\begin{equation}
\gamma = C \rho^{\epsilon} \ell_1^{-2} \gg \frac{6a}{\ell_0^3}
\end{equation}
if $\epsilon < \frac{1}{16}$.

Letting $\psi_0 = \ell_1^{-3}$,
\begin{eqnarray}
\nonumber E(2, \Lambda_{\ell_1}, V) &\geq& \frac{6a}{\ell_0^3} (1- \rho^{\epsilon}) \int_{\widetilde{\Lambda_{\ell_1}}^2} dx_1 dx_2 \ell_1^{-6} 1(|x_1 -x_2 | \leq \ell_0 ) - \frac{2 \|H_1 \|_{\infty}^2}{\gamma} \\
&\geq& \frac{6a}{\ell_0^3}(1-\rho^{\epsilon}) \frac{4 \pi \ell_0^3}{3} \frac{1}{\ell_1^3} = (1-\rho^{\epsilon}) \frac{8 \pi a }{\ell_1^3}
\end{eqnarray}
up to higher order terms of $\rho$. Note that
\begin{equation}
\frac{\|H_1 \|_{\infty}^2}{\gamma} \ll \frac{\rho^{\epsilon}}{\ell_1^3}
\end{equation}
if $\epsilon < \frac{1}{31}$. This proves Lemma \ref{lem5}.
\end{proof}

To prove Theorem \ref{thm1}, we divide $\mathbb{T}$ into small boxes of side length $\ell_2$, each of which is subdivided into smaller boxes of side length $\ell_1$. Let $\kappa = \rho^{-\epsilon}, \ell_2 = \ell_1 \kappa = \rho^{-\frac{1}{3} - \frac{2}{3} \epsilon}$. This means $\kappa^2 = N (\ell_2^3 / L^3 )$, i.e., the typical number of particles in $\Lambda_{\ell_2}$, a box of side length $\ell_2$, is $\kappa^2$. When the number of particles in a box is not too large, we can get the following lower bound.

\begin{lem} \label{lem6}
Let $V = V_1 - \lambda V_2$ and $a_1$ be the scattering length of $V_1$. Suppose that $k \leq M = \frac{a}{a_1} \cdot 8 \kappa^2 +1$. Given the Hamiltonian of $k$-particle system in a box of side length $\ell_2$
\begin{equation}
H^{(k)} = -\sum_{j=1}^k \Delta_j + \sum_{i<j}^k V(x_i -x_j), \label{hk}
\end{equation}
there exist positive constants $C'$ and $\epsilon$ such that, the ground state energy of (\ref{hk}) with Neumann boundary conditions,
\begin{equation}
E(k, \Lambda_{\ell_2}, V) \geq \frac{4 \pi a}{\ell_2^3} k(k-1) (1- C' \rho^{\epsilon}). \label{Ek}
\end{equation}
\end{lem}

\begin{proof}[Proof of Lemma \ref{lem6}.]
We have
\begin{eqnarray} \label{subdivision}
&& (1+5\rho^{\epsilon}) \Big[ (1-\rho^{\epsilon}) \sum_{j=1}^k \int_{\Lambda_{\ell_2}^k} dX_k |\nabla_j \psi(X_k)|^2 + \sum_{i<j}^k \int_{\Lambda_{\ell_2}^k} dX_k V(x_i -x_j) |\psi (X_k)|^2 \Big] \nonumber \\
&\geq& \Big[ \sum_{j=1}^k \int_{\Lambda_{\ell_2}^k} dX_k |\nabla_j \psi(X_k)|^2 + \sum_{i<j}^k \int_{\Lambda_{\ell_2}^k} dX_k (V_1 - (1-\rho^{\epsilon}) \lambda V_2 ) (x_i -x_j) |\psi (X_k)|^2 \Big] \\
&&+ 3\rho^{\epsilon} \Big[ \sum_{j=1}^k \int_{\Lambda_{\ell_2}^k} dX_k |\nabla_j \psi(X_k)|^2 + \sum_{i<j}^k \int_{\Lambda_{\ell_2}^k} dX_k (V_1 - 2\lambda V_2 ) (x_i -x_j) |\psi (X_k)|^2 \Big]. \nonumber
\end{eqnarray}

From Theorem \ref{thm2}, we know that the second term in the right hand side of (\ref{subdivision}) is non-negative. To estimate the first term in the right hand side of \eqref{subdivision}, subdivide $\Lambda_{\ell_2}$ into smaller boxes of side length $\ell_1$ as in the proof of Theorem \ref{thm2}: Consider a set $G = \{ p \in \mathbb{Z}^3 | \ell_1 p \in \Lambda_{\ell_2} \}$. Recall that $\chi$ is a characteristic function of the unit cube. We let $u \in [0, 1]^3 = \Gamma$ be the origin of the subdivision and, for $\lambda \in G$, $\chi_{\lambda u} (x) = \chi(x+\lambda + u)$. Let
\begin{equation}
\Lambda_{\lambda u} = \{x \in \Lambda_{\ell_2} | \chi_{\lambda u} (x/\ell_1) = 1 \}.
\end{equation}
$\Lambda_{\lambda u}$ is typically a box of side length of $\ell_1$. Note that some of the small boxes $\Lambda_{\lambda u}$ are smaller than others, in the sense that they are not the cubes of side length $\ell_1$. We want to consider the small boxes of side length $\ell_1$ only, which correspond to the small boxes in the subdivision that do not have walls in them in the proof of theorem \ref{thm2}. More precisely, we let
\begin{equation}
P(u) = \{ \Lambda_{\lambda u} : |\Lambda_{\lambda u}| = \ell_1^3 \}.
\end{equation}

Let
\begin{equation}
q(n)=
\begin{cases}
\displaystyle \frac{1-\rho^{\epsilon}}{1+5\rho^{\epsilon}} \cdot \frac{8 \pi a}{\ell_1^3} & \text{if}\;\; n=2 \\
0 & \text{otherwise} \\
\end{cases}.
\end{equation}
For a box of side length $\ell_1$, $\beta \in P(u)$, let
\begin{equation}
N_{\beta}(X_k) = \sum_{j=1}^k 1(x_j \in \beta).
\end{equation}
When $N_{\beta}(X_k) \neq 2$, Theorem \ref{thm2} shows that $E(N_{\beta}(X_k), \Lambda_{\ell_1} ,V) \geq 0$, and when $N_{\beta}(X_k)=2$, Lemma \ref{lem5} shows that $E(N_{\beta}(X_k), \Lambda_{\ell_1} ,V) \geq (1-\rho^{\epsilon}) \frac{8 \pi a}{\ell_1^3}$. Since $1-\rho^{\epsilon} \leq 1- \frac{\sqrt 3 R_1}{\ell_1} \leq h_{\ell_1}$,
\begin{eqnarray}
&& \sum_{j=1}^k \int_{\Lambda_{\ell_2}^k} dX_k |\nabla_j \psi(X_k)|^2 + \sum_{i<j}^k \int_{\Lambda_{\ell_2}^k} dX_k (V_1 - (1-\rho^{\epsilon}) \lambda V_2 ) (x_i -x_j) |\psi (X_k)|^2 \nonumber \\
&\geq& (1+5\rho^{\epsilon}) \int_{\Gamma} du \int_{\Lambda_{\ell_2}^k} dX_k \sum_{\beta \in P(u)} q(N_{\beta}(X_k)) |\psi (X_k)|^2.
\end{eqnarray}
Note that we could neglect small boxes that are not in $P(u)$, since for each of them the energy of the particles in it is non-negative as we have seen in Theorem \ref{thm2}. Together with \eqref{subdivision}, this shows that
\begin{equation}
H^{(k)} \geq H_1 = \rho^{\epsilon} \sum_{j=1}^k (-\Delta_j) + \int_{\Gamma} du \sum_{\beta \in P(u)} q(N_{\beta}(X_k))
\end{equation}
with Neumann boundary conditions.

For given $u \in [0, 1]^3$, to estimate $H_1$, we let 
\begin{equation}
W_u (X_k) = \sum_{\beta \in P(u)} q(N_{\beta}(X_k))
\end{equation}
be a perturbation, which depends on $u$. Since $W_u \leq \frac{k}{2} q(2)$ for all $u$ and all $X_k$, we have that $\gamma$, the spectral gap of the unperturbed Hamiltonian, $ \rho^{\epsilon} \sum_{j=1}^k (-\Delta_j)$, satisfies
\begin{equation}
\gamma = C \rho^{\epsilon} \ell_2^{-2} \gg C \kappa^2 \ell_1^{-3} \geq \|W_u\|_{\infty}
\end{equation}
if $\epsilon < \frac{1}{16}$. From Lemma \ref{perturb}, if we use the constant function $\ell_2^{-\frac{3}{2}k}$, the ground state of the unperturbed Hamiltonian, then
\begin{equation} \label{eqn_per2}
E(k, \Lambda_{\ell_2}, V) \geq \int_{\Gamma} du \Big( \ell_2^{-3k} \int_{\Lambda_{\ell_2}^k} dX_k W_u (X_k) - C \frac{\| W_u \|_{\infty}^2}{\rho^{\epsilon} \ell_2^{-2}} \Big).
\end{equation}

We compute the term inside the parenthesis in the right-hand side of \eqref{eqn_per2} for given $u$. To calculate the first term, we need to count the cases when exactly two particles are in given $\beta$. We use the inclusion-exclusion principle to get
\begin{eqnarray}
\nonumber && \ell_2^{-3k} \int_{\Lambda_{\ell_2}^k} dX_k W_u (X_k) \\
&\geq& \ell_2^{-3k} \int_{\Lambda_{\ell_2}^k} dX_k \sum_{\beta} \sum_{i_1<i_2}^k 1(x_{i_1} \in \beta) 1(x_{i_2} \in \beta) \frac{1- \rho^{\epsilon}}{1+ 5 \rho^{\epsilon}} \frac{8 \pi a}{\ell_1^3} \\
\nonumber && -3\; \ell_2^{-3k} \int_{\Lambda_{\ell_2}^k} dX_k \sum_{\beta} \sum_{i_1<i_2<i_3}^k 1(x_{i_1} \in \beta) 1(x_{i_2} \in \beta) 1(x_{i_3} \in \beta) \frac{1- \rho^{\epsilon}}{1+ 5 \rho^{\epsilon}} \frac{8 \pi a}{\ell_1^3}.
\end{eqnarray}
An explicit calculation shows
\begin{eqnarray}
&& \nonumber \ell_2^{-3k} \int_{\Lambda_{\ell_2}^k} dX_k W_u (X_k) \nonumber \\
&\geq& \frac{1- \rho^{\epsilon}}{1+ 5 \rho^{\epsilon}} \frac{8 \pi a}{\ell_1^3} \frac{k(k-1)}{2} \frac{(\ell_2 - 2\ell_1 )^3}{\ell_2^3} \frac{\ell_1^3}{\ell_2^3} - 3 \frac{1- \rho^{\epsilon}}{1+ 5 \rho^{\epsilon}} \frac{8 \pi a}{\ell_1^3} \frac{k(k-1)(k-2)}{6} \Big(\frac{\ell_1^3}{\ell_2^3} \Big)^2 \\
&\geq& \frac{4 \pi a}{\ell_2^3} k(k-1) (1- C \kappa^{-1}). \nonumber
\end{eqnarray}

We want the last term in \eqref{eqn_per2} is smaller than the error term in \eqref{Ek}, $C' \frac{4 \pi a}{{\ell_2^3}} k^2 \rho^{\epsilon}$, and this holds if $\epsilon < \frac{1}{22}$, since
\begin{equation}
C \frac{\|W\|_{\infty}^2}{\rho^{\epsilon} {\ell_2^{-2}}} \leq C \frac{k^2 \ell_2^2}{\rho^{\epsilon} \ell_1^6} \ll \frac{k^2 \kappa^{-1}}{\ell_2^3},
\end{equation}
provided $\epsilon < \frac{1}{22}$. Note that $a$ does not depend on $\rho$ hence could be absorbed into the constant $C$. This proves the lemma.
\end{proof}

When $k> M$, from Theorem \ref{thm2}, we have, since $V= V_1 - \frac{1}{2} c_1 V_2 = V_1 - \lambda V_2$,
\begin{eqnarray}
\nonumber && \sum_{j=1}^k \int_{\Lambda_{\ell_2}^k} dX_k |\nabla_j \psi(X_k)|^2 + \sum_{i<j}^k \int_{\Lambda_{\ell_2}^k} dX_k \; V(x_i-x_j) |\psi(X_k)|^2 \\
\nonumber &=& \frac{1}{2} \Big[ \sum_{j=1}^k \int_{\Lambda_{\ell_2}^k} dX_k |\nabla_j \psi(X_k)|^2 + \sum_{i<j}^k \int_{\Lambda_{\ell_2}^k} dX_k \; V_1(x_i-x_j) |\psi(X_k)|^2 \Big] \\
&& + \frac{1}{2} \Big[ \sum_{j=1}^k \int_{\Lambda_{\ell_2}^k} dX_k |\nabla_j \psi(X_k)|^2 + \sum_{i<j}^k \int_{\Lambda_{\ell_2}^k} dX_k \big( (V_1-c_1 V_2)(x_i-x_j) \big) |\psi(X_k)|^2 \Big] \label{k large 1}\\
\nonumber &\geq& \frac{1}{2} \Big[ \sum_{j=1}^k \int_{\Lambda_{\ell_2}^k} dX_k |\nabla_j \psi(X_k)|^2 + \sum_{i<j}^k \int_{\Lambda_{\ell_2}^k} dX_k \; V_1(x_i-x_j) |\psi(X_k)|^2 \Big] \\
\nonumber &\geq& \frac{1}{2} E(k, \Lambda_{\ell_2}, V_1 )
\end{eqnarray}

\vskip10pt

Now we prove the main theorem.

\begin{proof}[Proof of Theorem \ref{thm1}]
Define $\lambda$ by (\ref{delta}), (\ref{c1}), and (\ref{lambda}). We have
\begin{eqnarray}
\nonumber && \big(1+\frac{\sqrt 3 R_1}{\ell_2} \big) E(N, \mathbb{T}, V) \\
&=& \big(1+\frac{\sqrt 3 R_1}{\ell_2} \big) \inf_{\|\psi\|_2 =1} \Big[ \sum_{j=1}^N \int_{\mathbb{T}^N} dX_N |\nabla_j \psi(X_N)|^2 + \sum_{i<j}^N \int_{\mathbb{T}^N} dX_N V(x_i-x_j)|\psi(X_N)|^2 \Big] \\
\nonumber &=& \inf_{\|\psi\|_2 =1} \Big[ \sum_{j=1}^N \int_{\mathbb{T}^N} dX_N |\nabla_j \psi(X_N)|^2 + \sum_{i<j}^N \int_{\mathbb{T}^N} dX_N \big(V_1 - (1-\frac{\sqrt 3 R_1}{\ell_2}) \lambda V_2)(x_i-x_j) \big)|\psi(X_N)|^2 \\
\nonumber &&+\frac{\sqrt 3 R_1}{\ell_2} \Big( \sum_{j=1}^N \int_{\mathbb{T}^N} dX_N |\nabla_j \psi(X_N)|^2 + \sum_{i<j}^N \int_{\mathbb{T}^N} dX_N \big(V_1 - 2\lambda V_2) \big)|\psi(X_N)|^2 \Big) \Big].
\end{eqnarray}
Here, $R_1$ denotes the range of $V$. Theorem \ref{thm2} shows the last term in this equation satisfies
\begin{equation}
\sum_{j=1}^N \int_{\mathbb{T}^N} dX_N |\nabla_j \psi(X_N)|^2 + \sum_{i<j}^N \int_{\mathbb{T}^N} dX_N \big(V_1 - 2\lambda V_2) \big)|\psi(X_N)|^2 \geq 0.
\end{equation}
Hence,
\begin{equation}
\big(1+\frac{\sqrt 3 R_1}{\ell_2} \big) E(N, \mathbb{T}, V) \geq E\Big(N, \mathbb{T}, V_1 - (1-\frac{\sqrt 3 R_1}{\ell_2}) \lambda V_2 \Big) \geq E \Big( N, \mathbb{T}, \big(V_1 - \lambda V_2 ) \big)h_{\ell_2} \Big), \label{step1}
\end{equation}
where $h_{\ell_2}$ is defined by (\ref{hl}) using $\ell_2$ instead of $\ell$. Here, the second inequality follows from $1 \geq h_{\ell_2} \geq (1-\frac{\sqrt 3 R_1}{\ell_2})$. Subdivide $\mathbb{T}$ into smaller boxes of side length $\ell_2$ as in the proof of Theorem \ref{thm2} and index them by $\beta$. From Lemma \ref{lem7},
\begin{equation}
E(N, \mathbb{T}, Vh_{\ell_2}) \geq \inf_{\sum k_{\beta} =N} \sum_{\beta} E(k_{\beta}, \Lambda_{\ell_2}, V). \label{step2}
\end{equation}
Here, $k_{\beta}$ denotes the number of particles in $\beta$.

From Lemma \ref{lem6}, if $k \leq M$, we have
\begin{equation}
E(k, \Lambda_{\ell_2}, V) \geq \frac{4 \pi a}{\ell_2 ^3} k(k-1) (1- C' \rho^{\epsilon} ). \label{k small}
\end{equation}
From the superadditivity (\ref{eqn1}), we get, if $k_{\beta} > M$,
\begin{eqnarray}
\nonumber E(k_{\beta}, \Lambda_{\ell_2}, V_1 ) &\geq& \lfloor \frac{k_{\beta}}{M} \rfloor E(M, \Lambda_{\ell_2}, V_1) \geq \frac{k_{\beta}}{2M} E(M, \Lambda_{\ell_2}, V_1) \\
&\geq& \frac{4 \pi a_1}{\ell_2 ^3} (1- C' \rho^{\epsilon}) \frac{k_{\beta}}{2M} M(M-1) = \frac{4 \pi a}{\ell_2 ^3} (1- C' \rho^{\epsilon}) k_{\beta} \cdot 4 \kappa^2, \label{k large 2}
\end{eqnarray}
where $\lfloor \frac{k_{\beta}}{M} \rfloor$ denotes the greatest integer that does not exceed $\frac{k_{\beta}}{M}$.

Now we can actually calculate the lower bound using the argument similar to \cite{LY1}. From (\ref{k large 1}), (\ref{k small}), and (\ref{k large 2}), we get
\begin{equation}
\inf_{\sum k_{\beta} = N} \sum_{\beta} E(k_{\beta}, \Lambda_{\ell_2}, V) \geq \frac{4 \pi a}{\ell_2 ^3} (1- C' \rho^{\epsilon}) \inf_{\sum k_{\beta} = N} \Big[ \sum_{\beta: k_{\beta} \leq M} k_{\beta}(k_{\beta}-1) + \sum_{\beta: k_{\beta} > M} \frac{1}{2} k_{\beta} \cdot 4 \kappa^2 \Big]. \label{step3}
\end{equation}
Let $t= \sum_{k_{\beta} \leq M} k_{\beta}$. Then,
\begin{equation}
\sum_{\beta: k_{\beta} \leq M} k_{\beta}^2 \geq (\sum_{\beta: k_{\beta} \leq M} k_{\beta})^2 / (\sum_{\beta: k_{\beta} \leq M}1 ) \geq t^2 / (\frac{L^3}{\ell_2^3} ),
\end{equation}
hence $\sum k_{\beta} = N$ implies
\begin{equation}
\sum_{\beta: k_{\beta} \leq M} k_{\beta} (k_{\beta}-1) + \sum_{\beta: k_{\beta} > M} \frac{1}{2} k_{\beta} \cdot 4 \kappa^2 \geq \frac{\ell_2^3}{L^3} t^2 - t + 2 \kappa^2 (N-t) = \frac{\ell_2^3}{L^3} (t-N)^2 + N \kappa^2 -t. \label{min}
\end{equation}
Since $t \leq N$, the minimum of the right-hand side of (\ref{min}) is attained when $t=N$. Therefore, from (\ref{step1}), (\ref{step2}), (\ref{step3}), and (\ref{min}), there exists $C_0$ such that
\begin{equation}
E(N, \mathbb{T}, V) \geq (1 + \frac{\sqrt{3} R_1}{\ell_2} )^{-1} ( \frac{4 \pi a}{\ell_2^3}) (1- C' \rho^{\epsilon}) N (\kappa^2 -1) \geq 4 \pi a \rho N( 1- C_0 \rho^{\epsilon}),
\end{equation}
which was to be proved.
\end{proof}

\vskip20pt

\section*{Acknowledgements}
I am grateful to H. -T. Yau for
stimulating my interest in this problem, and for many helpful discussions.


\vskip30pt

\begin{thebibliography}{bib}

\bibitem[1]{CLY} Joseph G. Conlon, Elliott H. Lieb, Horng-Tzer Yau, \emph{The $N^{7/5}$ Law for Charged Bosons}, Commun. Math. Phys. \textbf{116}, 417-448 (1988).

\bibitem[2]{D} F. J. Dyson, \emph{Ground-State Energy of a Hard-Sphere Gas}, Phys. Rev. \textbf{106}, 20-26 (1957).

\bibitem[3]{ESY} Laszlo Erdos, Benjamin Schlein, Horng-Tzer Yau, \emph{Derivation of the Gross-Pitaevskii Hierarchy for the Dynamics of Bose-Einstein Condensate}, Commmun. Pure. Appl. Math. \textbf{59}, 1659-1741 (2005). arXiv:math-ph/0410005 v3 (2005).

\bibitem[4]{LSSY} E. H. Lieb, R. Seiringer, J. P. Solovej, J. Yngvason, \emph{The Ground State of the Bose Gas} in \emph{Current Developments in Mathematics, 2001}, International Press, Cambridge, 2002, pp. 131-178, arXiv:math-ph/0204027 v2 (2003).

\bibitem[5]{LSY} E. H. Lieb, R. Seiringer, J. Yngvason, \emph{Bosons in a Trap: A rigorous Derivation of the Gross-Pitaevskii Energy Functional}, Phys. Rev. A \textbf{61}, 043602-1 - 043602-13 (2000). arxiv:math-ph/9908027 v1 (1999).

\bibitem[6]{LY1} E. H. Lieb, J. Yngvason, \emph{Ground State of the Low Density Bose Gas}, Phys. Rev. Lett. \textbf{80}, 2504-2507 (1998). arXiv:cond-mat/9712138 v2 (1998).

\bibitem[7]{LY2} E. H. Lieb, J. Yngvason, \emph{The Ground State Energy of a Dilute Two-dimensional Bose Gas}, J. Stat. Phys. \textbf{103}, 509-526 (2001). arXiv:math-ph/0002014 v1 (2000).

\bibitem[8]{R} D. Ruelle, \emph{Statistical Mechanics: Rigorous Results}, New Ed edition, World Scientific Publishing Company (1999).

\bibitem[9]{Y} J. Yin, \emph{The Ground State Energy of Dilute Bose Gas in Potentials with Positive Scattering Length}, arXiv:math-ph/0808-4066 v1 (2008).

\end{thebibliography}
\end{document}